\documentclass[10pt]{article}
\usepackage{amsmath}
\usepackage{amssymb, amscd, amsthm}
\usepackage[all]{xy}
\usepackage[dvips]{graphicx}
\usepackage{verbatim}
\usepackage[perpage,symbol*]{footmisc}
\usepackage{cite}
\newtheorem{lemma}{\textbf{Lemma}}[section]
\newtheorem{theorem}{\textbf{Theorem}}
\newtheorem{remark}{\textbf{Remark}}[section]
\newtheorem{corollary}{\textbf{Corollary}}[section]

\newtheorem{example}{\textbf{Example}}[section]

\newtheorem{Definition}{\textbf{Definition}}
\usepackage{longtable}

\usepackage{multirow}

\newcommand{\F}{\mathbb{F}}

\begin{document}

\baselineskip 17pt\title{\Large\bf On Euclidean Hulls of MDS Codes}

\author{\large  Xiaolei Fang \quad\quad Meiqing Liu \quad\quad Jinquan Luo*}\footnotetext{The authors are with School of Mathematics
and Statistics \& Hubei Key Laboratory of Mathematical Sciences, Central China Normal University, Wuhan China, 430079. \\ E-mail: fangxiaolei@mails.ccnu.edu.cn(X.Fang), 15732155720@163.com(M.Liu), luojinquan@mail.ccnu.edu.cn(J.Luo)}

\date{}
\maketitle

{\bf Abstract}: In this paper, we propose a mechanism on the constructions of MDS codes with arbitrary dimensions of Euclidean hulls.
Precisely, we construct (extended) generalized Reed-Solomon(GRS) codes with assigned dimensions of Euclidean hulls from self-orthogonal GRS codes.
It turns out that our constructions are more general than previous works on Euclidean hulls of (extended) GRS codes.

{\bf Key words}: MDS self-orthogonal code, Generalized Reed-Solomon(GRS) code, Extended generalized Reed-Solomon(GRS) code, Euclidean hull

\section{Introduction}

 \quad\; Let $q$ be a prime power and $\mathbb{F}_{q}$ a finite field with $q$ elements. A $q$-ary $[n,k,d]$ code $C$ is a linear code over $\mathbb{F}_{q}$
with length $n$, dimension $k$ and minimum distance $d$. The Singleton bound states that $k\leq n-d+1$. The code $C$ attaching the Singleton bound(i.e., $k=n-d+1$) is called a maximum distance separable(MDS) code. Due to their optimal properties, MDS codes play an important role in coding theory and related fields, see [\ref{BR}, \ref{SR}].

For any two vectors $\overrightarrow{a}=(a_{1},a_{2},\ldots,a_{n})$ and $\overrightarrow{b}=(b_{1},b_{2},\ldots,b_{n})\in \mathbb{F}_{q}^{n}$,
we define their Euclidean inner product as:
\begin{equation*}
\overrightarrow{a} \cdot \overrightarrow{b}=\sum\limits_{i=1}^{n}a_{i}b_{i}.
\end{equation*}
The dual code of $C$ is defined as
\begin{equation*}
C^\perp=\left\{\overrightarrow{a} \in \mathbb{F}_{q}^{n}:\overrightarrow{a} \cdot \overrightarrow{b}=0  \text{ for any }\overrightarrow{b} \in C \right\}.
\end{equation*}
The hull of $C$ is defined by $$Hull(C)=C\cap C^\perp.$$ Readers are referred to [\ref{AK}] for more details on hull of linear code.

The code $C$ satisfying $Hull(C)=\left\{ \overrightarrow{0} \right\}$ is called a linear complementary dual(LCD) code, which has been extensively
investigated recently ([\ref{CG}, \ref{CMTQ}, \ref{CMTQP}, \ref{LDL}]). In [\ref{CG}], Carlet et al. investigated constructions of LCD codes utilizing cyclic codes, expanded Reed-Solomon codes and generalized residue codes, together with direct sum, puncturing, shortening, extension, $(u|u+v)$ construction and suitable automorphism action. In [\ref{CMTQ}] and [\ref{CMTQP}], Carlet et al. showed that any linear code over $\mathbb{F}_{q}$ $(q>3)$ is equivalent to a Euclidean LCD code and any
linear code over $\mathbb{F}_{q^{2}}$ $(q>2)$ is equivalent to a Hermitian LCD code. In [\ref{LDL}], Li et al. presented some LCD cyclic codes with very
good parameters in general and a well-rounded treatment of reversible cyclic codes is also given. The linear code $C$ satisfying $Hull(C)=C$
(resp. $C^\perp$) is called self-orthogonal (resp. dual containing) code. In particular, the code $C$ satisfying $C=C^\perp$ is called a self-dual code.
Some MDS self-dual codes are constructed through various ways, see [\ref{FF3}, \ref{GKL}, \ref{JX2}, \ref{LLL}, \ref{TW}, \ref{Yan}, \ref{ZF}].
On the other hand, many classes of quantum MDS codes are constructed by MDS Hermitian self-orthogonal codes,
see [\ref{FF1}, \ref{FF2}, \ref{HXC}-\ref{JX1}, \ref{LXW}, \ref{SYC}, \ref{SYZ}, \ref{ZC}, \ref{ZG}].

In general, linear codes with assigned dimensions of hulls can be applied to construct entanglement-assisted quantum error-correcting codes(EAQECCs).
EAQECCs were firstly introduced in [\ref{BDH}]. Wilde and Brun proposed a method for constructing EAQECCs by utilizing classical linear codes over finite
fields [\ref{WB}]. However, it is not explicit to calculate the required number of entangled states. Recently, this number is related to the hull
of classical linear code [\ref{GJG}]. Thereafter, several
new families of optimal EAQECCs are proposed by determining the hulls of classical linear codes, see [\ref{FFLZ}, \ref{GJG}, \ref{LC}, \ref{LCC}].

By using (extended) GRS codes, Luo et al. proposed several infinite families of MDS codes with hulls of arbitrary dimensions, which can be applied to
construct some families of MDS EAQECCs with flexible parameters, see [\ref{LC}] and [\ref{LCC}]. In [\ref{FFLZ}], Fang et al. presented several MDS codes
by utilizing (extended) GRS codes, and determined the dimensions of their Euclidean or Hermitian hulls. In particular, some of the associated EAQECCs have
the required number of maximally entangled states. They also gave some new classes of MDS codes with Hermitian hulls of arbitrary dimensions.

Based on [\ref{FFLZ}], [\ref{LC}] and [\ref{LCC}], we propose a mechanism on the constructions of MDS codes with arbitrary dimensions of Euclidean hulls.
After the main results, we give some examples.

The rest of this paper is organized as follows. In Section 2, we briefly recall some basic notations and properties of (extended) GRS codes. In Section 3,
the mechanism on general constructions of MDS codes with Euclidean hulls of arbitrary dimensions is presented. We give several examples to illustrate the
general construction mechanism in Section 4. Section 5 concludes the paper.

\section{Preliminaries}

 \quad\; In this section, we introduce some basic notations and useful results on (extended) GRS codes.
Readers are referred to [\ref{MS}, Chapter 10] for more details.

Let $\mathbb{F}_{q}$ be a finite field with $q$ elements. Denote by $\F_q^*=\F_q\backslash\{0\}$. For $1\leq n\leq q$, choose two vectors
$\overrightarrow{v}=(v_{1},v_{2},\ldots,v_{n})\in (\mathbb{F}_{q}^{*})^{n}$ and $\overrightarrow{a}=(a_{1},a_{2},\ldots,a_{n})\in \mathbb{F}_{q}^{n}$,
where $a_{i}$ $(1\leq i\leq n)$ are distinct. For an integer $k$ with $1\leq k\leq n$, the GRS code of length $n$ associated with $\overrightarrow{v}$
and $\overrightarrow{a}$ is defined as follows:
\begin{equation}\label{def GRS}
\mathbf{GRS}_{k}(\overrightarrow{a},\overrightarrow{v})=\left\{(v_{1}f(a_{1}),\ldots,v_{n}f(a_{n})):f(x)\in\mathbb{F}_{q}[x], \mathrm{deg}(f(x))\leq k-1\right\}.
\end{equation}
A generator matrix of $\mathbf{GRS}_{k}(\overrightarrow{a},\overrightarrow{v})$ is
\begin{equation*}
G_{k}=\left(
\begin{array}{cccc}
v_{1}&v_{2}&\cdots&v_{n}\\
v_{1}a_{1}&v_{2}a_{2}&\cdots&v_{n}a_{n}\\
v_{1}a_{1}^{2}&v_{2}a_{2}^{2}&\cdots&v_{n}a_{n}^{2}\\
\vdots&\vdots&\ddots&\vdots\\
v_{1}a_{1}^{k-1}&v_{2}a_{2}^{k-1}&\cdots&v_{n}a_{n}^{k-1}\\
\end{array}
\right).
\end{equation*}
The code $\mathbf{GRS}_{k}(\overrightarrow{a},\overrightarrow{v})$ is a $q$-ary $[n,k]$ MDS code and its dual is also MDS [\ref{MS}, Chapter 11].

The extended GRS code associated with $\overrightarrow{v}$ and $\overrightarrow{a}$ is defined by:
\begin{equation}\label{def extended GRS}
\mathbf{GRS}_{k}(\overrightarrow{a},\overrightarrow{v},\infty)=\{(v_{1}f(a_{1}),\ldots,v_{n}f(a_{n}),f_{k-1}):f(x)\in\mathbb{F}_{q}[x],
\mathrm{deg}(f(x))\leq k-1\},
\end{equation}
where $f_{k-1}$ is the coefficient of $x^{k-1}$ in $f(x)$. A generator matrix of $\mathbf{GRS}_{k}(\overrightarrow{a},\overrightarrow{v},\infty)$ is
\begin{equation*}
G_{k,\infty}=\left(
\begin{array}{ccccc}
v_{1}&v_{2}&\cdots&v_{n}&0\\
v_{1}a_{1}&v_{2}a_{2}&\cdots&v_{n}a_{n}&0\\
v_{1}a_{1}^{2}&v_{2}a_{2}^{2}&\cdots&v_{n}a_{n}^{2}&0\\
\vdots&\vdots&\ddots&\vdots&\vdots\\
v_{1}a_{1}^{k-1}&v_{2}a_{2}^{k-1}&\cdots&v_{n}a_{n}^{k-1}&1\\
\end{array}
\right).
\end{equation*}
The code $\mathbf{GRS}_{k}(\overrightarrow{a},\overrightarrow{v},\infty)$ is a $q$-ary $[n+1,k]$ MDS code and its dual is also MDS [\ref{MS}, Chapter 11].

For $1\leq i \leq n$, we define
\begin{equation}\label{solution}
u_i:=\prod_{1\leq j\leq n,j\neq i}(a_{i}-a_{j})^{-1}.
\end{equation}
Let $QR_q$ denote the set of nonzero square elements of $\F_q$. These symbols will be used frequently in this paper.

\begin{lemma}([\ref{CL}, Lemma 2])\label{y1}
A codeword $\overrightarrow{c}=(v_{1}f(a_{1}),\ldots,v_{n}f(a_{n}))\in Hull\big(\mathbf{GRS}_{k}(\overrightarrow{a},\overrightarrow{v})\big)$ if and only if there exists a polynomial $g(x)\in\mathbb{F}_{q}[x]$ with
$\mathrm{deg}(g(x))\leq n-k-1$, such that
\begin{center}
$(v_{1}^{2}f(a_{1}),v_{2}^{2}f(a_{2}),\ldots,v_{n}^{2}f(a_{n}))=(u_{1}g(a_{1}),u_{2}g(a_{2}),\ldots,u_{n}g(a_{n}))$.
\end{center}
\end{lemma}

\begin{lemma}([\ref{CL}, Lemma 3])\label{y2}
A codeword $\overrightarrow{c}=(v_{1}f(a_{1}),\ldots,v_{n}f(a_{n}),f_{k-1}) \in Hull\big(\mathbf{GRS}_{k}(\overrightarrow{a},\overrightarrow{v},\infty)\big)$ if and only if there exists a polynomial $g(x)\in\mathbb{F}_{q}[x]$
with $\mathrm{deg}(g(x))\leq n-k$, such that
\begin{center}
$(v_{1}^{2}f(a_{1}),v_{2}^{2}f(a_{2}),\ldots,v_{n}^{2}f(a_{n}),f_{k-1})=(u_{1}g(a_{1}),u_{2}g(a_{2}),\ldots,u_{n}g(a_{n}),-g_{n-k})$.
\end{center}
\end{lemma}

\begin{lemma}([\ref{LXW}, Lemma 5])\label{y3}
Let $a_1,a_2,\cdots,a_n$ be distinct elements in $\mathbb{F}_{q}$. Then we have
\begin{equation*}
\sum_{i=1}^na_i^mu_i=
\begin{cases}
0, \text{ $0\leq m\leq n-2$;}\\
1, \text{ $ m=n-1$.}
\end{cases}
\end{equation*}
\end{lemma}

In Corollary 2.4 of [\ref{JX2}], sufficient condition for GRS codes being self-dual is presented. In the following lemma, we show that the condition
is also necessary. Furthermore, an equivalent condition for a GRS code being self-orthogonal is presented. 

\begin{lemma}\label{GRS}
If $1\leq m\leq \lfloor\frac{n}{2}\rfloor$, then $\mathbf{GRS}_m(\overrightarrow{a},\overrightarrow{v})$ is Euclidean self-orthogonal if and only if
$v_i^2=\lambda(a_i)u_i\neq0(1\leq i\leq n)$, where $\lambda(a_i)=\lambda_0+\lambda_1 a_i+\cdots+\lambda_{n-2m}a_i^{n-2m}$ with
$\lambda_h\in\F_q(0\leq h\leq n-2m)$.
\end{lemma}
\begin{proof}
It is easy to check that
\begin{equation*}
\begin{aligned}
\mathbf{GRS}_{m}(\overrightarrow{a},\overrightarrow{v})\,\text{is self-orthogonal}
\Leftrightarrow \sum\limits_{i=1}^n v_{i}^{2}a_{i}^{l}=0 \,\,\text{for}\,\, 0\leq l \leq 2m-2.
\end{aligned}
\end{equation*}
Denote by $x_i=v_i^2$($1\leq i\leq n$). The system of linear equations
\begin{equation}\label{sum}
\begin{aligned}
\sum\limits_{i=1}^n a_{i}^{l}x_i=0
\end{aligned}
\end{equation}
for $0\leq l\leq 2m-2$ has solutions
\begin{equation}\label{AX=0}
(u_1,\ldots,u_n),  (a_1 u_1,\ldots,a_n u_n),\ldots, (a_1^{n-2m}u_1,\ldots,a_n^{n-2m}u_n),
\end{equation}
which are linear independent. Note that the rank of coefficient matrix of (\ref{sum}) is $2m-1$.  It follows that (\ref{AX=0}) is a basic
solution system of (\ref{sum}). Therefore,
\begin{center}
$v_i^2=\sum\limits_{h=0}^{n-2m}\lambda_h a_i^hu_i\neq 0\,\,\text{for any}\,\,1\leq i\leq n\,\,\text{and}\,\,\lambda_h\in \F_q$.
\end{center}

Conversely, let $v_i^2=\lambda(a_i)u_i\neq0(1\leq i\leq n)$ where $\lambda(a_i)=\lambda_0+\lambda_1 a_i+\cdots+\lambda_{n-2m}a_i^{n-2m}$ with
$\lambda_{h}\in\F_q(0\leq h\leq n-2m)$. Then
\begin{equation*}
\sum\limits_{i=1}^n v_i^{2}a_{i}^l=0 \ \text{for}\ 0\leq l \leq 2m-2.
\end{equation*}
It implies $\mathbf{GRS}_m(\overrightarrow{a},\overrightarrow{v})$ is Euclidean self-orthogonal.
\end{proof}

\begin{corollary}\label{cor1}
Assume $1\leq m\leq \lfloor\frac{n}{2}\rfloor$. Then
$\mathbf{GRS}_{m}(\overrightarrow{a},\overrightarrow{v})^\perp=\mathbf{GRS}_{n-m}(\overrightarrow{a},\overrightarrow{v})$ if and only if there exists
$\lambda\in \mathbb{F}_{q}^{*}$  such that $\lambda u_{i}=v_{i}^{2}$, where $1\leq i\leq n$. In particular, when $m=\frac{n}{2}$ with $n$ even,
$\mathbf{GRS}_{\frac{n}{2}}(\overrightarrow{a},\overrightarrow{v})$ is MDS self-dual (see Corollary 2.4 of [\ref{JX2}]).
\end{corollary}

Similarly as GRS codes, Lemma 2 of [\ref{Yan}] presents sufficient condition for extended GRS codes being self-dual. The following
lemma shows that the condition is also necessary. More precisely, we give a criterion for an extended GRS code being self-orthogonal.

\begin{lemma}\label{eGRS}
If $1\leq m\leq \lfloor\frac{n+1}{2}\rfloor$, then $\mathbf{GRS}_m(\overrightarrow{a},\overrightarrow{v},\infty)$ is Euclidean self-orthogonal
if and only if $v_i^2=\lambda(a_i)u_i\neq0(1\leq i\leq n)$, where $\lambda(a_i)=\lambda_0+\lambda_1 a_i+\cdots+\lambda_{n-2m}a_i^{n-2m}-a_i^{n-2m+1}$
with $\lambda_h\in\F_q(0\leq h\leq n-2m)$.
\end{lemma}

\begin{proof}
By taking inner product of all pairs in the basis of $\mathbf{GRS}_m(\overrightarrow{a},\overrightarrow{v},\infty)$,
\begin{equation*}
\begin{aligned}
\mathbf{GRS}_m(\overrightarrow{a},\overrightarrow{v},\infty)\,\text{is self-orthogonal}
\Leftrightarrow \begin{cases}
\sum\limits_{i=1}^n v_i^{2}a_{i}^l=0 ,&0\leq l \leq 2m-3;\\
\sum\limits_{i=1}^n v_i^{2}a_{i}^{2m-2}+1=0.
\end{cases}
\end{aligned}
\end{equation*}
Denote by $x_i=v_i^2(1\leq i\leq n)$. If we only consider the system of equations $\sum\limits_{i=1}^n a_{i}^lx_i=0 (0\leq l \leq 2m-3)$,
similarly as Lemma \ref{GRS}, the solution is
\begin{equation}\label{vsquare}
x_i=\sum\limits_{h=0}^{n-2m+1}\lambda_h a_i^hu_i\neq 0\,\,\text{for any}\,\,1\leq i\leq n\,\, \text{and}\,\,\lambda_h\in \F_q.
\end{equation}
Substituting (\ref{vsquare}) to $\sum\limits_{i=1}^n a_{i}^{2m-2}x_i+1=0$,
\begin{equation*}
\sum\limits_{i=1}^n\lambda_{n-2m+1}a_i^{n-1}u_i+1=0.
\end{equation*}
It deduces that $\lambda_{n-2m+1}=-1$ from Lemma \ref{y3}. Hence $v_i^2=\lambda(a_i)u_i(1\leq i\leq n)$ where
$\lambda(a_i)=\lambda_0+\lambda_1 a_i+\cdots+\lambda_{n-2m}a_i^{n-2m}-a_i^{n-2m+1}$ with $\lambda_h\in\F_q(0\leq h\leq n-2m)$.

Conversely, let $v_i^2=\lambda(a_i)u_i\neq0$ for any $1\leq i\leq n$ and
$\lambda(a_i)=\lambda_0+\lambda_1 a_i+\cdots+\lambda_{n-2m}a_i^{n-2m}-a_i^{n-2m+1}$ with $\lambda_h\in\F_q(0\leq h\leq n-2m)$. Then
\begin{equation*}
\begin{cases}
\sum\limits_{i=1}^n v_i^{2}a_{i}^l=0 ,&0\leq l \leq 2m-3\\
\sum\limits_{i=1}^n v_i^{2}a_{i}^{2m-2}+1=0.
\end{cases}
\end{equation*}
This completes the proof.
\end{proof}

\begin{corollary}\label{cor2}
For $1\leq m\leq \lfloor\frac{n+1}{2}\rfloor$, the code
$\mathbf{GRS}_{m}(\overrightarrow{a},\overrightarrow{v},\infty)^{\perp}=\mathbf{GRS}_{n+1-m}(\overrightarrow{a},\overrightarrow{v},\infty)$ if and only if
$-u_{i}=v_{i}^{2}$ for all $i=1,2 ,\ldots,n$. In particular, when $n$ is odd and $m=\frac{n+1}{2}$,
$\mathbf{GRS}_{\frac{n+1}{2}}(\overrightarrow{a},\overrightarrow{v})$ is MDS self-dual (see Lemma 2.2 of [\ref{Yan}]).
\end{corollary}

\section{Main Results}

 \quad\; In this section, we present our constructions of MDS codes with Euclidean hulls of arbitrary dimensions utilizing (extended) GRS codes.

Firstly, we give the definition of almost self-dual code. It is a special case of self-orthogonal code.

\begin{Definition}\label{almost self-dual}
Assume the length of the code $C$ is odd. If $C\subseteq C^\perp$ and $\dim(C^\perp)=\dim(C)+1$, we call $C$ an almost self-dual code.
\end{Definition}

Now we construct MDS codes with Euclidean hulls of arbitrary dimensions via GRS codes.
\begin{theorem}\label{via GRS}
Assume $1\leq m\leq \lfloor\frac{n}{2}\rfloor$ and $q>3$. Suppose $\mathbf{GRS}_{m}(\overrightarrow{a},\overrightarrow{v})$ is self-orthogonal
(i.e. $\mathbf{GRS}_{m}(\overrightarrow{a},\overrightarrow{v})\subseteq \mathbf{GRS}_{m}(\overrightarrow{a},\overrightarrow{v})^\perp$)
with $\overrightarrow{a}=(a_{1},a_{2},\ldots,a_{n})$ and $\overrightarrow{v}=(v_{1},v_{2},\ldots,v_{n})$. For any
$0\leq l\leq k\leq m \leq \lfloor\frac{n}{2}\rfloor$, there exists a $q$-ary $[n,k]$ MDS code $C$ with $\dim(Hull(C))=l$.
\end{theorem}
\begin{proof}
From $\mathbf{GRS}_{m}(\overrightarrow{a},\overrightarrow{v})\subseteq \mathbf{GRS}_{m}(\overrightarrow{a},\overrightarrow{v})^\perp$ and Lemma \ref{GRS},
\begin{center}
$v_i^2=\lambda(a_i)u_i\neq 0(1\leq i\leq n),$
\end{center}
where $\lambda(a_i)=\lambda_0+\lambda_1a_i+\ldots+\lambda_{n-2m}a_i^{n-2m}$ with $\lambda_{h}\in \mathbb{F}_{q}(0\leq h\leq n-2m)$.
Denote by $s:=k-l$, $\overrightarrow{a}=(a_{1},a_{2},\ldots,a_{n})$ and
$\overrightarrow{v}^{'}=(\alpha v_{1},\alpha v_{2},\ldots,\alpha v_{s},v_{s+1},\ldots,v_{n})$, where $\alpha\in\mathbb{F}_{q}^*$ and $\alpha^2\neq 1$.
For $C=\mathbf{GRS}_k(\overrightarrow{a},\overrightarrow{v}^{'})$ and any
\begin{center}
$\overrightarrow{c}=(\alpha v_{1}f(a_{1}),\ldots,\alpha v_{s}f(a_{s}),v_{s+1}f(a_{s+1}),\ldots,v_{n}f(a_{n}))\in Hull(C)$
\end{center}
with $\deg(f(x))\leq k-1$, according to Lemma \ref{y1}, there exists a polynomial $g(x)\in \mathbb{F}_{q}[x]$ with $\deg(g(x))\leq n-k-1$ such that
\begin{equation*}
\begin{aligned}
&(\alpha^{2} v_{1}^{2}f(a_{1}),\ldots,\alpha^{2} v_{s}^{2}f(a_{s}),v_{s+1}^{2}f(a_{s+1}),\ldots,v_{n}^{2}f(a_{n}))&\\
=&(u_{1}g(a_{1}),\ldots,u_{s}g(a_{s}),u_{s+1}g(a_{s+1}),\ldots,u_{n}g(a_{n})).&
\end{aligned}
\end{equation*}
Due to $v_i^2=\lambda(a_i)u_i$($1\leq i\leq n$),
\begin{equation}\label{equ1}
\begin{aligned}
&(\alpha^{2}\lambda(a_1)u_{1}f(a_{1}),\ldots,\alpha^{2}\lambda(a_s)u_{s}f(a_{s}),\lambda(a_{s+1})u_{s+1}f(a_{s+1}),\ldots,\lambda(a_n)u_{n}f(a_{n}))&\\
=&(u_{1}g(a_{1}),\ldots,u_{s}g(a_{s}),u_{s+1}g(a_{s+1}),\ldots,u_{n}g(a_{n})).&
\end{aligned}
\end{equation}
When $s+1 \leq  i \leq  n$, we get $\lambda(a_i)f(a_i)=g(a_i)$. Note that $\deg(\lambda(x)f(x))\leq n-2m+(k-1)\leq n-2k+(k-1)=n-k-1$ and
$\deg(g(x))\leq n-k-1$. It deduces that $\lambda(x)f(x)=g(x)$ from $n-s\geq n-k$. When $1\leq i\leq s$, it implies
\begin{center}
$\alpha^{2}\lambda(a_i)u_{i}f(a_{i})=u_{i}g(a_{i})=u_{i}\lambda(a_i)f(a_{i})$.
\end{center}
We derive that $f(a_{i})=0$ ($1\leq i \leq s$) by $\alpha^{2}\neq 1$ and $\lambda(a_i) u_{i}\neq 0$. So
\begin{equation*}
f(x)=h(x)\prod_{i=1}^{s}(x-a_{i}),
\end{equation*}
for some $h(x)\in \mathbb{F}_{q}[x]$ with $\deg(h(x))\leq k-1-s$. It follows that $\dim(Hull(C))\leq k-s$.

Conversely, put $f(x)=h(x)\prod\limits_{i=1}^{s}(x-a_{i})$, where $h(x)\in \mathbb{F}_{q}[x]$ and $\deg(h(x))\leq k-1-s$.
Assume that $g(x)=\lambda(x)f(x)$, which yields $\deg(g(x))\leq n-k-1$. Then
\begin{equation*}
\begin{aligned}
&(\alpha^{2}\lambda(a_1)u_{1}f(a_{1}),\ldots,\alpha^{2}\lambda(a_s)u_{s}f(a_{s}),\lambda(a_{s+1})u_{s+1}f(a_{s+1}),\ldots,\lambda(a_n)u_{n}f(a_{n}))&\\
=&(u_{1}g(a_{1}),\ldots,u_{s}g(a_{s}),u_{s+1}g(a_{s+1}),\ldots,u_{n}g(a_{n})).&
\end{aligned}
\end{equation*}
According to Lemma \ref{y1},
\begin{center}
$(\alpha v_{1}f(a_{1}),\ldots,\alpha v_{s}f(a_{s}), v_{s+1}f(a_{s+1}),\ldots,  v_{n}f(a_{n}))\in Hull(C)$.
\end{center}
Therefore, $\dim(Hull(C))\geq k-s$.

As a result, $\dim(Hull(C))= k-s=l$.
\end{proof}

As a corollary, the following result can be deduced by choosing
$\mathbf{GRS}_{m}(\overrightarrow{a},\overrightarrow{v})^\perp=\mathbf{GRS}_{n-m}(\overrightarrow{a},\overrightarrow{v})$.
\begin{corollary}\label{cor GRS}
Assume $1\leq m\leq \lfloor\frac{n}{2}\rfloor$ and $q>3$. Suppose
$$\mathbf{GRS}_{m}(\overrightarrow{a},\overrightarrow{v})^\perp=\mathbf{GRS}_{n-m}(\overrightarrow{a},\overrightarrow{v})$$
with $\overrightarrow{a}=(a_{1},a_{2},\ldots,a_{n})$ and $\overrightarrow{v}=(v_{1},v_{2},\ldots,v_{n})$. For any
$0\leq l\leq k\leq \lfloor\frac{n}{2}\rfloor$, there exists a $q$-ary $[n,k]$ MDS code $C$ with $\dim(Hull(C))=l$.
\end{corollary}

\begin{remark}
Both Theorem 7 of [\ref{LCC}] and Theorem 1(i) of [\ref{FFLZ}] are special cases of Corollary \ref{cor GRS}.
\end{remark}

The above result is on the constructions of MDS codes with Euclidean hulls of arbitrary dimensions utilizing GRS codes. Afterwards, we present
constructions utilizing extended GRS codes.

\begin{theorem}\label{via eGRS}
Assume $1\leq m \leq \lfloor\frac{n+1}{2}\rfloor$, $q>3$ and $n<q$. Suppose \,$\mathbf{GRS}_{m}(\overrightarrow{a},\overrightarrow{v},\infty)$ is self-orthogonal with $\overrightarrow{a}=(a_{1},a_{2},\ldots,a_{n})$ and $\overrightarrow{v}=(v_{1},v_{2},\ldots,v_{n})$. For any
$0\leq l\leq k\leq m \leq \lfloor\frac{n+1}{2}\rfloor$, there exists a $q$-ary $[n+1,k]$ MDS code $C$ with $\dim(Hull(C))=l$.
\end{theorem}

\begin{proof}
Since $\mathbf{GRS}_{m}(\overrightarrow{a},\overrightarrow{v},\infty)$ is self-orthogonal and by Lemma \ref{eGRS},
\begin{center}
$v_i^2=\lambda(a_i)u_i\neq 0(1\leq i\leq n),$
\end{center}
where $\lambda(a_i)=\lambda_0+\lambda_1a_i+\ldots+\lambda_{n-2m}a_i^{n-2m}-a_i^{n-2m+1}$ with $\lambda_{h}\in \mathbb{F}_{q}(0\leq h\leq n-2m)$.
Put $\pi(x)=(x-b)^{m-k}$ with some $b\in\F_q\backslash\{a_1,\ldots,a_n\}$. Denote by $s:=k-l$. Choose
\begin{center}
$\overrightarrow{a}=(a_1,\ldots,a_n)$ and
$\overrightarrow{v}^{'}=(\alpha v_{1}\pi(a_1),\alpha v_{2}\pi(a_2),\ldots,\alpha v_{s}\pi(a_s),v_{s+1}\pi(a_{s+1}),\ldots,v_{n}\pi(a_n))$,
\end{center}
where $\alpha\in\mathbb{F}_{q}^*$ with $\alpha^{2}\neq 1$. Set $C:=\mathbf{GRS}_k(\overrightarrow{a},\overrightarrow{v}^{'},\infty)$.
For any
\begin{center}
$\overrightarrow{c}=(\alpha v_{1}\pi(a_1)f(a_{1}),\ldots,\alpha v_{s}\pi(a_s)f(a_{s}),v_{s+1}\pi(a_{s+1})f(a_{s+1}), \ldots,v_{n}\pi(a_n)f(a_{n}),f_{k-1})\in Hull(C)$
\end{center}
with $\deg(f(x))\leq k-1$, by Lemma \ref{y2}, there exists a polynomial $g(x)\in \mathbb{F}_{q}[x]$ with $\deg(g(x))\leq n-k$ such that
\begin{equation*}
\begin{aligned}
&(\alpha^{2} v_{1}^{2}\pi^2(a_1)f(a_{1}),\ldots,\alpha^{2} v_{s}^{2}\pi^2(a_s)f(a_{s}),v_{s+1}^{2}\pi^2(a_{s+1})f(a_{s+1}),\ldots,v_{n}^{2}\pi^2(a_{n})f(a_{n}),f_{k-1})&\\
=&(u_{1}g(a_{1}),\ldots,u_{s}g(a_{s}),u_{s+1}g(a_{s+1}),\ldots,u_{n}g(a_{n}),-g_{n-k}).&
\end{aligned}
\end{equation*}
From $v_i^2=\lambda(a_i)u_i$, we derive
\begin{equation}\label{equ2}
\begin{aligned}
&(\alpha^{2}\lambda(a_1)u_{1}\pi^2(a_1)f(a_{1}),\ldots,\alpha^{2}\lambda(a_s)u_{s}\pi^2(a_s)f(a_{s}),\lambda(a_{s+1})u_{s+1}\pi^2(a_{s+1})f(a_{s+1}),\ldots, &\\ &\lambda(a_n)u_{n}\pi^2(a_{n})f(a_{n}),f_{k-1})=(u_{1}g(a_{1}),\ldots,u_{s}g(a_{s}),u_{s+1}g(a_{s+1}),\ldots,u_{n}g(a_{n}),-g_{n-k}).&
\end{aligned}
\end{equation}
We claim that $\lambda(x)\pi^2(x)f(x)=g(x)$ in the following:
\begin{itemize}
  \item {\bf Case 1:} $-f_{k-1}=g_{n-k}=0$. It follows from (\ref{equ2}) that $\lambda(a_i)\pi^2(a_i)f(a_i)=g(a_i)$ for $s+1\leq i \leq n$. Note that
  $\deg(\lambda(x)\pi^2(x)f(x))\leq n-2m+1+2m-2k+k-2=n-k-1$ and $\deg(g(x))\leq n-k-1$. From $n-s\geq n-k$, it follows that $\lambda(x)\pi^2(x)f(x)=g(x)$.
  \item {\bf Case 2:} $-f_{k-1}=g_{n-k}\neq 0$. In this case, $\deg(\lambda(x)\pi^2(x)f(x))=n-2m+1+2m-2k+k-1=n-k$ and $\deg(g(x))=n-k$. Then
  $\deg(\lambda(x)\pi^2(x)f(x)-g(x))\leq n-k-1$. From (\ref{equ2}), $\lambda(a_i)\pi^2(a_i)f(a_i)=g(a_i)$ for $s+1\leq i \leq n$. Since $n-s\geq n-k$,
  then $\lambda(x)\pi^2(x)f(x)=g(x)$.
\end{itemize}
Comparing the beginning $s$ coordinates on both sides of (\ref{equ2}),
\begin{center}
$\alpha^{2}\lambda(a_i)u_{i}\pi^2(a_i)f(a_{i})=u_{i}g(a_{i})=u_{i}\lambda(a_i)\pi^2(a_i)f(a_{i})$.
\end{center}
We derive that $f(a_{i})=0$ ($1\leq i \leq s$) by $\alpha^{2}\neq 1$ and $\lambda(a_i) u_{i}\pi(a_i)\neq 0$. So
\begin{equation*}
f(x)=h(x)\prod_{i=1}^{s}(x-a_{i}),
\end{equation*}
for some $h(x)\in \mathbb{F}_{q}[x]$ with $\deg(h(x))\leq k-s-1$. It follows that $\dim(Hull(C))\leq k-s$.

Conversely, set $f(x)=h(x)\prod\limits_{i=1}^{s}(x-a_{i})$, where $h(x)\in \mathbb{F}_{q}[x]$ and $\deg(h(x))\leq k-1-s$. Assume that
$g(x)=\lambda(x)\pi^2(x)f(x)$, which implies $\deg(g(x))\leq n-k$. Then
\begin{equation*}
\begin{aligned}
&(\alpha^{2}\lambda(a_1)u_{1}\pi^2(a_1)f(a_{1}),\ldots,\alpha^{2}\lambda(a_s)u_{s}\pi^2(a_s)f(a_{s}),\lambda(a_{s+1})u_{s+1}\pi^2(a_{s+1})f(a_{s+1}),\ldots,&\\
&\lambda(a_n)u_{n}\pi^2(a_n)f(a_{n}),f_{k-1})=(u_{1}g(a_{1}),\ldots,u_{s}g(a_{s}),u_{s+1}g(a_{s+1}),\ldots,u_{n}g(a_{n}),-g_{n-k}).&
\end{aligned}
\end{equation*}
According to Lemma \ref{y2},
\begin{center}
$(\alpha v_{1}\pi(a_1)f(a_{1}),\ldots,\alpha v_{s}\pi(a_s)f(a_{s}), v_{s+1}\pi(a_{s+1})f(a_{s+1}),\ldots,  v_{n}\pi(a_n)f(a_{n}),f_{k-1})\in Hull(C)$.
\end{center}
Therefore, $\dim(Hull(C))\geq k-s$.

Consequently, $\dim(Hull(C))= k-s=l$.
\end{proof}

As a corollary of this theorem, the following result can be derived directly by choosing self-dual code
$\mathbf{GRS}_{\frac{n+1}{2}}(\overrightarrow{a},\overrightarrow{v},\infty)$ with $n$ odd
(self-dual code $\mathbf{GRS}_{\frac{n}{2}}(\overrightarrow{a},\overrightarrow{v})$ with $n$ even, respectively).
\begin{corollary}\label{cor eGRS}
(i). Assume $n$ is odd, $q>3$ and $n<q$. Suppose \,$\mathbf{GRS}_{\frac{n+1}{2}}(\overrightarrow{a},\overrightarrow{v},\infty)$ is self-dual with
$\overrightarrow{a}=(a_{1},a_{2},\ldots,a_{n})$ and $\overrightarrow{v}=(v_{1},v_{2},\ldots,v_{n})$ . For any $0\leq l\leq k \leq \frac{n+1}{2}$,
 there exists a $q$-ary $[n+1,k]$ MDS code $C$ with $\dim(Hull(C))=l$.

(ii). Assume $n$ is even and $q>3$. Let $\mathbf{GRS}_{\frac{n}{2}}(\overrightarrow{a},\overrightarrow{v})$ be self-dual with
$\overrightarrow{a}=(a_{1},a_{2},\ldots,a_{n})$ and $\overrightarrow{v}=(v_{1},v_{2},\ldots,v_{n})$ . For any $1\leq k \leq \frac{n}{2}$ and
$0\leq l\leq k-1$, there exists a $q$-ary $[n+1,k]$ MDS code $C$ with $\dim(Hull(C))=l$.
\end{corollary}

\begin{remark}
As special cases of this result, Theorem 1(ii),(iii) and Theorem 2 of [\ref{FFLZ}] can be deduced directly from Corollary \ref{cor eGRS}.
\end{remark}

The remaining case $q=3$ can be depicted explicitly.
\begin{remark}\label{q=3}

(i). The $3$-ary $[2,1,2]$ MDS code $C$ with generator matrix
\begin{center}
$G_1=\left(
  \begin{array}{cc}
    v_1 & v_2 \\
  \end{array}
\right)$
\end{center}
where $v_1,v_2\in\F_3^*$, has $\dim(Hull(C))=0$.

(ii).  The $3$-ary $[3,1,3]$ MDS code $C$ with generator matrix
\begin{center}
$G'_1=\left(
  \begin{array}{ccc}
    v_1 & v_2 & v_3 \\
  \end{array}
\right)$
\end{center}
where $v_1,v_2,v_3\in\F_3^*$, has $\dim(Hull(C))=1$.

(iii). The $3$-ary $[4,1,4]$ MDS code $C$ with generator matrix
\begin{center}
$G''_{1,\infty}=\left(
  \begin{array}{cccc}
    v_1 & v_2 & v_3 & 1 \\
  \end{array}
\right)$
\end{center}
where $v_1,v_2,v_3\in\F_3^*$, has $\dim(Hull(C))=0$ and the $3$-ary $[4,2,3]$ MDS code $C$ with generator matrix
\begin{center}
$G''_{2,\infty}=\left(
  \begin{array}{cccc}
    v_1 & v_2 & v_3 & 0 \\
    0 & v_2 & -v_3 & 1 \\
  \end{array}
\right)$
\end{center}
where $v_1,v_2,v_3\in\F_3^*$, has $\dim(Hull(C))=2$.  A straightforward calculation shows that there does not exist $3$-ary $[4,2,3]$ code $C$ with $\dim(Hull(C))=1$.
\end{remark}

\section{Examples}

\quad\; Each MDS self-orthogonal (extended) GRS code can be applied to construct MDS codes with arbitrary dimensions of hulls. In this section,
applying Theorems \ref{via GRS} and \ref{via eGRS}, we give some concrete examples on (extended) GRS codes whose dimensions of hulls can be determined.

\begin{example}\label{even}
Let $q=r^2$, where $r$ is an odd prime power. Suppose $m\mid q-1$. For $1\leq t\leq \frac{r+1}{\gcd(r+1,m)}$, assume $n=tm$ is even.

(i). If $\frac{q-1}{m}$ is even, then for any $1\leq k \leq \frac{n}{2}$ and \,$0\leq l\leq k$, there exists a $q$-ary $[n,k]$ MDS code $C$
with $\dim(Hull(C))=l$.

(ii). If $\frac{q-1}{m}$ is even, then for any $1\leq k \leq \frac{n-1}{2}$ and \,$0\leq l\leq k-1$, there exists a $q$-ary $[n+1,k]$ MDS code $C$
with $\dim(Hull(C))=l$.

(iii). For any $1\leq k \leq \frac{n}{2}$ and \,$0\leq l\leq k$, there exists a $q$-ary $[n+1,k]$ MDS code $C$ with $\dim(Hull(C))=l$,
except the case that $t$ is even, $m$ is even and $r\equiv1\,(\mathrm{mod}\,4)$.

(iv). For any $1\leq k \leq \frac{n+2}{2}$ and \,$0\leq l\leq k$, there exists a $q$-ary $[n+2,k]$ MDS code $C$ with $\dim(Hull(C))=l$,
except the case that $t$ is even, $m$ is even and $r\equiv1\,(\mathrm{mod}\,4)$.
\end{example}

\begin{proof}
(i). Let $\alpha$ be a primitive $m$-th root of unity in $\F_q$ and $S=\langle\beta\rangle$ be the cyclic group of order $r+1$. By the second fundamental
theorem of group homomorphism,
$$S\big/(S\cap\langle\alpha\rangle)\simeq(S\times\langle\alpha\rangle)\big/\langle\alpha\rangle\leq\F_q^*\big/\langle\alpha\rangle.$$
Let $B=\{\beta^{\mu_1},\ldots,\beta^{\mu_t}\}$ be a set of coset representatives of $(S\times\langle\alpha\rangle)\big/\langle\alpha\rangle$ with
$0\leq \mu_1<\cdots<\mu_t<r+1$. Put $\mu=\mu_1+\cdots+\mu_t$ and
$A=\{\alpha\beta^{\mu_1},\ldots,\alpha^m\beta^{\mu_1},\alpha\beta^{\mu_2},\ldots,\alpha^m\beta^{\mu_2},\ldots,\alpha\beta^{\mu_t},
\ldots,\alpha^m\beta^{\mu_t}\}$. Denote by $a_{c+(j-1)m}:=\alpha^c\beta^{\mu_j}$ with $1\leq c\leq m$, $1\leq j\leq t$ and
$\overrightarrow{a}=(a_{1},\ldots,a_{n})$. Let $i=c+(j-1)m$ and $\lambda=g^{\frac{r+1}{2}\cdot(t-1)-m\mu}$, where $1\leq i\leq n$ and $g$ is a
generator of $\F_q^*$. Then by [\ref{FLLL}], we know $\lambda\cdot\prod\limits_{z\neq i,\,\,z=1}^n\left(a_i-a_z\right)\in QR_q$.
Set $v_i^2=\left(\lambda\cdot\prod\limits_{z\neq i,\,\,z=1}^n\left(a_i-a_z\right)\right)^{-1}$ and
$\overrightarrow{v}=(v_1,\ldots,v_n)$. Then $\mathbf{GRS}_{\frac{n}{2}}(\overrightarrow{a},\overrightarrow{v})$ is MDS self-dual.
According to Theorem \ref{via GRS}, we complete the proof.

(ii). With the same process of proof as (i) and Theorem \ref{via eGRS}, we can obtain the result.

(iii). Similarly as (i), choose
$A=\{\alpha\beta^{\mu_1},\ldots,\alpha^m\beta^{\mu_1},\alpha\beta^{\mu_2},\ldots,\alpha^m\beta^{\mu_2},\ldots,\alpha\beta^{\mu_t},
\ldots,\alpha^m\beta^{\mu_t},0\}.$ Denote by $a_{c+(j-1)m}:=\alpha^c\beta^{\mu_j}$, $a_{n+1}:=0$ and
$\overrightarrow{a}=(a_1,\ldots,a_n,a_{n+1})$, where $1\leq c\leq m$ and $1\leq j\leq t$.
Let $i=c+(j-1)m$ $(1\leq i\leq n)$. Then by [\ref{FLLL}], we have $\prod\limits_{z\neq i,\,\,z=1}^{n+1}\left(a_i-a_z\right)\in QR_q$, for any
$1\leq i\leq n+1$, except the case that $t$ is even, $m$ is even and $r\equiv1\,(\mathrm{mod}\,4)$. Accordingly, for any
$1\leq i\leq n+1$, we can set $v_i^2=\prod\limits_{z\neq i,\,\,z=1}^n\left(a_i-a_z\right)^{-1}$ and
$\overrightarrow{v}=(v_1,\ldots,v_n,v_{n+1})$. It follows that $\mathbf{GRS}_{\frac{n-1}{2}}(\overrightarrow{a},\overrightarrow{v})$ is MDS almost
self-dual. Due to Theorem \ref{via GRS}, the result can be deduced.

(iv). With the same process of proof as (iii), we let $\overrightarrow{a}=(a_1,\ldots,a_n,a_{n+1})$ and $\overrightarrow{v}=(v_1,\ldots,v_n,v_{n+1})$,
where $v_i^2=-\prod\limits_{z\neq i,\,\,z=1}^n\left(a_i-a_z\right)^{-1}$. Since
$\mathbf{GRS}_{\frac{n+1}{2}}(\overrightarrow{a},\overrightarrow{v},\infty)$ is MDS self-dual and by Theorem \ref{via eGRS}, we obtain the result.
\end{proof}

\begin{remark}
In (ii) and (iii), the length of the code is $n+1$. However, they can not cover each other.
\end{remark}

\begin{example}\label{odd}
Let $q=r^2$, where $r$ is an odd prime power. Suppose $m\mid q-1$ and $1\leq t\leq \frac{r+1}{2\gcd(r+1,m)}$. Assume $n=tm$ is odd.

(i). For any $1\leq k \leq \frac{n-1}{2}$ and $0\leq l\leq k$, there exists a $q$-ary $[n,k]$ MDS code $C$ with $\dim(Hull(C))=l$.

(ii). For any $1\leq k \leq \frac{n+1}{2}$ and $0\leq l\leq k$, there exists a $q$-ary $[n+1,k]$ MDS code $C$ with $\dim(Hull(C))=l$.

(iii). For any $1\leq k \leq \frac{n+1}{2}$ and $0\leq l\leq k-1$, there exists a $q$-ary $[n+2,k]$ MDS code $C$ with $\dim(Hull(C))=l$.
\end{example}
\begin{proof}
(i). Recall $\alpha$ and $\beta$ in the proof of Example \ref{even}. Let $B=\{\beta^{\mu_1},\ldots,\beta^{\mu_t}\}$ be a set of coset representatives of
$(S\times\langle\alpha\rangle)\big/\langle\alpha\rangle$ with $0\leq\mu_1<\cdots<\mu_t<r+1$ and $\mu_1,\ldots,\mu_t$ are even. Denote by
$\mu=\mu_1+\cdots+\mu_t$ and
$A=\{\alpha\beta^{\mu_1},\ldots,\alpha^m\beta^{\mu_1},\alpha\beta^{\mu_2},\ldots,\alpha^m\beta^{\mu_2},\ldots,\alpha\beta^{\mu_t},
\ldots,\alpha^m\beta^{\mu_t}\}$. Put $a_{c+(j-1)m}:=\alpha^c\beta^{\mu_j}$ with $1\leq c\leq m$, $1\leq j\leq t$ and
$\overrightarrow{a}=(a_1,\ldots,a_n)$. Let $i=c+(j-1)m$ with $1\leq i\leq n$. Then by [\ref{FLLL}], we derive that
\begin{center}
$\prod\limits_{z\neq i,\,\,z=1}^n\left(a_i-a_z\right)\in QR_q$.
\end{center}
Let $v_i^2=\prod\limits_{z\neq i,\,\,z=1}^n\left(a_i-a_z\right)^{-1}$ and $\overrightarrow{v}=(v_1,\ldots,v_n)$. It yields
$\mathbf{GRS}_{\frac{n-1}{2}}(\overrightarrow{a},\overrightarrow{v})$ is MDS almost self-dual. By Theorem \ref{via GRS}, we finish the proof.

(ii). With the same process as (i), let $\overrightarrow{a}=(a_1,\ldots,a_n)$ and we obtain $\prod\limits_{z\neq i,\,\,z=1}^n\left(a_i-a_z\right)\in QR_q$
by [\ref{FLLL}]. Hence there exists $v_i\in \mathbb{F}_q^*$ so that $v_i^2=-\prod\limits_{z\neq i,\,\,z=1}^n\left(a_i-a_z\right)^{-1}$. It is easy to see that
$\mathbf{GRS}_{\frac{n+1}{2}}(\overrightarrow{a},\overrightarrow{v},\infty)$ is MDS self-dual. Then the result follows from Theorem \ref{via eGRS}.

(iii). Choose $A=\{\alpha\beta^{\mu_1},\ldots,\alpha^m\beta^{\mu_1},\alpha\beta^{\mu_2},\ldots,\alpha^m\beta^{\mu_2},\ldots,\alpha\beta^{\mu_t},
\ldots,\alpha^m\beta^{\mu_t},0\}$. Denote by
\begin{center}
$a_{c+(j-1)m}:=\alpha^c\beta^{\mu_j}$, $a_{n+1}:=0$ and $\overrightarrow{a}=(a_1,\ldots,a_n,a_{n+1})$,
\end{center}
where $1\leq c\leq m$ and $1\leq j\leq t$.
Let $i=c+(j-1)m$ $(1\leq i\leq n)$. Then by [\ref{FLLL}], we deduce that $\prod\limits_{z\neq i,\,\,z=1}^{n+1}\left(a_i-a_z\right)\in QR_q$, for any
$1\leq i\leq n+1$. Thus we let $v_i^2=-\prod\limits_{z\neq i,\,\,z=1}^{n+1}\left(a_i-a_z\right)^{-1}$ ($1\leq i\leq n+1$) and
$\overrightarrow{v}=(v_1,\ldots,v_n,v_{n+1})$. Then the result follows from that $\mathbf{GRS}_{\frac{n}{2}}(\overrightarrow{a},\overrightarrow{v})$
is MDS self-dual and Theorem \ref{via eGRS}.
\end{proof}

\begin{example}\label{additive}
Let $q=p^{2s}$, where $p$ is an odd prime and $s$ is a positive integer. Assume that $n=p^{2e}$ with $1\leq e\leq s$.

(i). For any $1\leq k \leq \frac{n-1}{2}$ and $0\leq l\leq k$, there exists a $q$-ary $[n,k]$ MDS code $C$ with $\dim(Hull(C))=l$.

(ii). For any $1\leq k \leq \frac{n+1}{2}$ and $0\leq l\leq k$, there exists a $q$-ary $[n+1,k]$ MDS code $C$ with $\dim(Hull(C))=l$.
\end{example}
\begin{proof}
(i). Denote by $r=p^{s}$. Let $S=\{\alpha_1,\alpha_2,\ldots,\alpha_{p^e}\}$ be an $e$-dimensional $\F_p$-linear subspace of $\F_r$ with $1\leq e\leq s$.
Choose $\beta\in\F_q\backslash\F_r$ such that $\beta^{r+1}=1$. Let $\alpha_{k,j}=\alpha_k\beta+\alpha_j$ with $1\leq k,j\leq p^e$. Denote by
$a_{k+(j-1)p^e}:=\alpha_{k,j}$ and $\overrightarrow{a}=(a_1,\ldots,a_n)$. Let $i=k_0+(j_0-1)\cdot p^e$ with $1\leq i\leq n$. Then by [\ref{FLLL}],
it follows that $\prod\limits_{z\neq i,\,\,z=1}^{n}\left(a_i-a_z\right)\in QR_q$. For any $1\leq i\leq n$, set
$v_i^2=\prod\limits_{z\neq i,\,\,z=1}^{n}\left(a_i-a_z\right)^{-1}$ and $\overrightarrow{v}=(v_1,\ldots,v_n)$. It is easy to see that
$\mathbf{GRS}_{\frac{n-1}{2}}(\overrightarrow{a},\overrightarrow{v})$ is MDS almost self-dual. According to Theorem \ref{via GRS},
we accomplish the proof.

(ii). With the same reason as (i), put $\overrightarrow{a}=(a_1,\ldots,a_n)$. We obtain $\prod\limits_{z\neq i,\,\,z=1}^{n}\left(a_i-a_z\right)\in QR_q$
with $1\leq i\leq n$. Let $v_i^2=-\prod\limits_{z\neq i,\,\,z=1}^{n}\left(a_i-a_z\right)^{-1}$ and denote by $\overrightarrow{v}=(v_1,\ldots,v_n)$. We
deduce that $\mathbf{GRS}_{\frac{n+1}{2}}(\overrightarrow{a},\overrightarrow{v},\infty)$ is MDS self-dual. According to Theorem \ref{via eGRS},
the result can be obtained.
\end{proof}

From Theorem 6.1 in [\ref{ZF}], when $q\equiv3\,(\mathrm{mod}\,4)$ and $n\equiv2\,(\mathrm{mod}\,4)$, there does not exist self-dual code over $\F_q$ with
length $n$. However, self-orthogonal codes with $q\equiv3\,(\mathrm{mod}\,4)$ and $n\equiv2\,(\mathrm{mod}\,4)$ may exist. So we can construct MDS codes
with Euclidean hulls of assigned dimensions with $q\equiv3\,(\mathrm{mod}\,4)$ and $n\equiv2\,(\mathrm{mod}\,4)$ by Theorem \ref{via GRS} in the following example.

\begin{example}
Let $q\equiv3\,(\mathrm{mod}\,4)$ be an odd prime power. Suppose odd $t\mid q-1$ and $n=2t$. For any $1\leq k\leq \frac{n}{2}-1$ and
$0\leq l\leq k$, there exists a $q$-ary $[n,k]$ MDS code $C$ with $\dim(Hull(C))=l$.
\end{example}
\begin{proof}
Let $\alpha$ be a primitive $t$-th root of unity in $\F_q$. For any $\omega\not\in QR_q$, set
$$\overrightarrow{a}=\left(\alpha,\alpha^2,\ldots,\alpha^t,\omega\alpha,\omega\alpha^2,\ldots,\omega\alpha^t\right).$$
 When $1\leq i\leq t$,
\begin{equation*}
\begin{aligned}
u_i=\prod\limits_{j=1,j\neq i}^t(\alpha^i-\alpha^j)\cdot\prod\limits_{j=1}^t(\alpha^i-\omega\alpha^j)=t\alpha^{-i}\cdot(1-\omega^t)
\end{aligned}
\end{equation*}
and
\begin{equation*}
\begin{aligned}
u_{i+t}=\prod\limits_{j=1,j\neq i}^t(\omega\alpha^i-\omega\alpha^j)\cdot\prod\limits_{j=1}^t(\omega\alpha^i-\alpha^j)=(-\omega^{t-1})\cdot t\alpha^{-i}\cdot(1-\omega^t).
\end{aligned}
\end{equation*}
Choose $\lambda(x)=t(1-\omega^t)x$.  For $1\leq i\leq t$,
$$\lambda(\alpha^i)u_i=(t\cdot(1-\omega^t))^2\in QR_q$$
and
$$\lambda(\omega\alpha^{i})u_{i+t}=(-\omega^t)\cdot(t\cdot(1-\omega^t))^2\in QR_q,$$
which follows from $q\equiv3\,(\mathrm{mod}\,4)$ and $t$ odd. By Lemma \ref{GRS}, there exists $\overrightarrow{v}\in \mathbb{F}_q^n$ with nonzero entries such that
$\mathbf{GRS}_k(\overrightarrow{a},\overrightarrow{v})$ is self-orthogonal. According to Theorem \ref{via GRS}, we complete the proof.
\end{proof}

\section{Conclusion}

\quad\;Based on [\ref{FFLZ}], [\ref{LC}] and [\ref{LCC}], we propose a mechanism on the constructions of MDS codes with arbitrary dimensions of Euclidean
hulls: if there exist self-orthogonal (extended) GRS codes, then we can construct (extended) GRS codes with arbitrary assigned dimensions of Euclidean hulls.
 In particular, MDS (almost) self-dual codes can be employed to construct such codes. In this sense, any known (extended) GRS (almost) self-dual code can be applied to
 find new (extended) GRS code  with any dimension of hull. A more general question remains open: for an $[n,m]$ MDS code $C$ with $\dim\left(Hull(C)\right)=h$, try to find $[n,k]$ MDS code $C'$ with  any $k\leq m$ and any $\dim\left(Hull(C')\right)=l\leq \min{(h,k)}$. We invite readers to attack this open problem.

\section*{Acknowledgements}

The authors thank anonymous reviewers and editor for their suggestions and comments to improve the readability of this paper. This research is supported by National Natural Science Foundation of China under Grant 11471008, Grant 11871025
and the self-determined research funds of CCNU from the colleges' basic research and operation of MOE(Grant No. CCNU18TS028).

\end{document}